\documentclass[12pt,a4paper,twoside]{amsart}

\usepackage{graphicx}
\usepackage{amsfonts}
\usepackage{amssymb}
\usepackage{amsthm}
\usepackage[centertags]{amsmath}
\usepackage{newlfont}
\usepackage[normalem]{ulem}
\usepackage{subcaption}
\usepackage{hyperref}

\usepackage{color}
\usepackage{graphicx,color}
\usepackage{amsmath, amssymb, graphics}
\usepackage[pagewise]{lineno}

\newtheorem{theorem}{Theorem}[section]
\newtheorem{proposition}[theorem]{Proposition}

\theoremstyle{definition}

\newtheorem{remark}[theorem]{Remark}

\begin{document}
 
\title[An application of of circular hodograph theorem]{A small variation of the circular hodograph theorem and the best elliptical trajectory of the planets}

\author{Carman Cater, Oscar Perdomo, Amanda Valentine}
\address{Department of Mathematics\\ Central Connecticut State University\\ New Britain, CT 06050, USA}

\date{}

\begin{abstract} 
A small variation of the circular shape of the hodograph theorem states that for every elliptical solution of the two-body problem, it is possible to find an appropriate inertial frame such that the speed of the bodies is constant. We use this result and data from the NASA JPL Horizon Web Interface to find the best fitting ellipse for the trajectory of Mercury, Venus, Earth, Mars, and Jupiter.  The process requires us to find procedures to obtain the plane and ellipse that best fit a collection of points in space. We show that if we aim for the plane that minimizes the sum of the square distances from the given points to the unknown plane, we obtain three planes that appear  to divide the set of points equally into octants, one of these being our desired plane of best fit. We provide a detailed proof of the hodograph theorem.
\end{abstract}
 
\maketitle

\section{Introduction} 

Kepler's Second Law states that for a planet traveling around the sun, the segment connecting the sun with the planet sweeps out equal areas in equal times.  A similar relation is given by the circular shape of the hodograph theorem.  This theorem applies not only to planets orbiting a sun, but also comets following conical orbits. It states that for any conical solution of the two-body problem, when we consider the velocity vectors of the motion of one of the two bodies the tips of these vectors lies on a circle when we fix their tails to a fixed point. See Figure \ref{fig:VV1}

\begin{figure}[h!]
  \includegraphics[width=.4\linewidth]{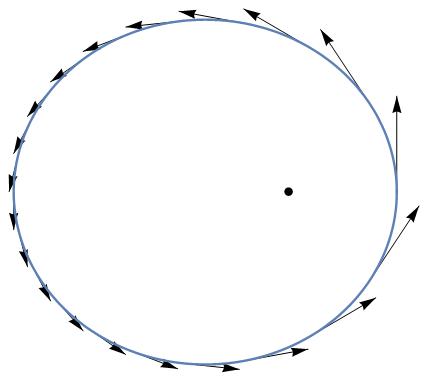}\includegraphics[width=.3\linewidth]{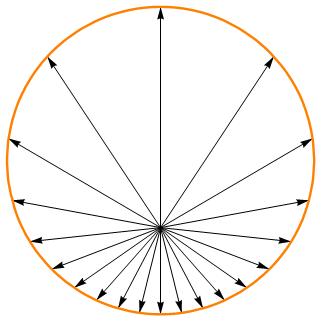} 
  \caption{  Left: Trajectory of a planet around the sun with its velocity vectors.  Right: Velocity vectors with tails fixed at the origin.}
  \label{fig:VV1}
\end{figure}



As mentioned in {\color{blue}\cite{B}} according to Goldstein {\color{blue} \cite{G}}, for the elliptical orbits this theorem was first communicated in 1846 by Hamilton.

In this paper we point out (and for completeness we provide a proof) a small variation of the circular shape of the hodograph theorem.  It states that after an appropriate inertial change of coordinates the fixed point for the tail of the vectors can be made to be the center of the circle, and therefore, the length of each velocity vector --the speed-- is constant under this new system of reference. When motion is given by ellipses, the direction of this suitable inertial frame is perpendicular to the major axis of the ellipse.  

The small variation of the circular shape of the hodograph theorem provides not only the direction of the major axis of the ellipse, but the eccentricity as well. To illustrate an application we use this fact to compute the best ellipse that fits the trajectory of some planets.  Even though in this paper we explain a procedure to find the best ellipse by just using the points from the trajectory, we decided to use use the hodograph theorem so that not only the points but also the velocity vectors are used in our choice of the best ellipse.

In the process of finding the best ellipse we need to compute the best plane that fits a collection of points. For sake of completeness we show how to find the plane that minimizes the square of the distances from the given points to the unknown plane. An interesting observation that we noticed is that this process not only gives us a best plane, but it also provides three planes that evenly distribute the points.  For greater generality we show this for any hyperplane in the $ n $ dimensional Euclidean space $\mathbb{R}^n$.  

Section \ref{tbp} explains and provides a proof of a small variation of the circular hodograph theorem.  Also, in the case an elliptical motion it provides formulas for the parameters involved in the circular hodograph theorem in terms of the parameters of the ellipse. Section \ref{bp} describes how to find the best hyperplane that fits a collection of points in $\mathbb{R}^n$, and also explains the observation that this procedure gives us $ n $ hyperplanes that evenly distribute the points.  Section \ref{be} describes the procedure that we are using to find the best ellipse that fits a given collection of points. We do two procedures, one with no assumption on the ellipse, and the other assumes that we know the direction of one of the axes as well as the eccentricity. We also explain how to take the velocities and positions of the planet from NASA JPL Horizon Web Interface. Section \ref{p} gives the step by step procedure we use for each planet.  Lastly, Section \ref{r} shows the parametrizations of the ellipse that best fits the trajectory of Mercury, Venus, Earth, Mars, and Jupiter.

\section{A small variation of the circular hodograph theorem}\label{tbp}

Since the proof of the hodograph theorem is obtained by just adding a few lines to the proof of the fact that the trajectories for the two-body problem are conics, this section not only provides the proof of the hodograph theorem but also explains the solution of the two-body problem.  

\subsection{Proof of the hodograph theorem}Let us consider the two-body problem and assume that the masses of the two bodies are $M$ and $m$.  Furthermore, let us assume that the body with mass $M$ moves with position $x(t)=(x_1(t),x_2(t))$ and the body with mass $m$ moves with position $y(t)=(y_1(t),y_2(t))$. 
If we let $r=y-x$ and $\mu=(M+m)G$, where $G$ is the gravitational constant, we get that $\ddot{r}= -\frac{(M+m) G}{|r|^3}  \,  r= -\frac{\mu \, r}{|r|^3}$.  Assuming that the trajectory of the bodies are not in a line and their center of mass stays fixed, then these trajectories are ellipses, parabolas, or hyperbolas. Taking $r(0)$ the closest point to the origin, and defining $e_1=-\frac{r(0)}{|r(0)|}$ and  $e_2=-\frac{\dot{r}(0)}{|\dot{r}(0)|}$ we have that these vectors form an orthonormal basis and 
we can find functions $\rho$ and $\theta$ such that  $r=\rho\cos(\theta) e_1+\rho\sin(\theta) e_2$. Letting $|r(0)|=r_0>0$ we have that $r(0)=-r_0 e_1=r_0 \cos(\pi)e_1+r_0 \sin(\pi)e_2$.  Therefore we can pick $\theta(0)=\pi$ and  $\rho(0)=r_0$. Our choice of $r(0)$ gives us that that $\dot{\rho}(0)=0$, therefore if $|\dot{r}(0)|=v$, then  $-ve_2 =  \dot{r}(0)= \dot{\theta}(0)\rho(0)\cos(\pi) e_2 $, which implies that  $\dot{\theta}(0)=\frac{v}{r_0}$. Let us define the unit vector fields $u_1$ and $u_2$ by 
$$u_1(t)=\cos(\theta(t)) e_1 +\sin(\theta(t)))e_2 \quad\hbox{and}\quad u_2(t)=-\sin(\theta(t)) e_1 +\cos(\theta(t)))e_2 $$
then $r(t)=\rho(t)u_1(t)$, $\dot{r}= \dot{\rho}u_1+\dot{\theta} \rho u_2$ and,
$$\ddot{r}= (\ddot{\rho}-\dot{\theta}^2\rho)u_1+(2\dot{\theta} \dot{\rho}+\rho\ddot{\theta})u_2=
 (\ddot{\rho}-\dot{\theta}^2\rho)u_1+\frac{1}{\rho}\,\frac{d({\dot{\theta}\rho^2})}{dt}u_2$$
Since $\ddot{r}=-\frac{\mu \, r}{|r|^3}= -\frac{\mu}{\rho^2}\, u_1$ then, $\frac{d({\dot{\theta}\rho^2})}{dt}$ must be zero and 

\begin{eqnarray}\label{sie}
 \ddot{\rho}-\frac{k^2}{\rho^3}+\frac{\mu }{\rho^2}=0
\end{eqnarray}

where $k=\dot{\theta}\rho^2$ which we know is constant because  $\frac{d({\dot{\theta}\rho^2})}{dt}=0$.  From the previous Equation (\ref{sie}) we have that

\begin{eqnarray} \label{cc}
 (\dot{\rho})^2+\frac{k^2}{\rho^2}-2 \frac{\mu }{\rho}=c
\end{eqnarray}

for some constant $c$. Since $\rho(0)=r_0$, $\dot{\rho}(0)=0$,  and $\dot{\theta}(0)=\frac{v}{r_0}$, then $k= vr_0$ and $c=v^2-\frac{2\mu}{r_0}$.  Since $\rho=\frac{a}{1-e \cos(\theta)}$ describes the equation of a conic in polar coordinates $(\rho,\theta)$ with eccentricity $e$, we now look for a solution of Equation (\ref{cc}) of the form

$$\rho(t)=\frac{A}{1-B \cos(\theta(t))}$$

Note that replacing $t=0$ gives $r_0=\frac{A}{1+B}$. Since 

$$\dot{\theta}=\frac{k}{\rho^2}=\frac{k(1-B\cos(\theta(t)))^2}{A^2}$$

then 

$$\dot{\rho}=-\frac{Bk}{A}\sin(\theta)$$

and Equation (\ref{cc}) reduces to 

$$\frac{-2 B \left(k^2-A \mu \right) \cos (\theta (t))-2 A \mu +\left(B^2+1\right) k^2}{A^2}=c$$

A direct computation shows that making $A=\frac{k^2}{\mu}=\frac{v^2r_0^2}{\mu}$ and $B=\frac{r_0v^2-\mu}{\mu}$ solves the equation above and also satisfies the equation $r_0=\frac{A}{1+B}$. We conclude that  $r(t)$ lies in a conic with eccentricity $e=|\frac{r_0v^2-\mu}{\mu}|$. 
So far we have shown that the trajectories of the solutions of the two-body problem are conics.  The variation of the hodograph theorem is obtained by noticing if
$$\xi=\frac{r_0 v^2-\mu }{r_0 v}$$
and $V=\dot{r}+\xi e_2$, then 
$$V\cdot V=\frac{\mu ^2}{r_0^2 v^2}$$

Let us check the computation in detail.

\begin{proposition}
Assuming the notation used in this section, if we take the vector $V=\dot{r}+\xi e_2$ and let $ \xi=\frac{r_0 v^2-\mu }{r_0 v} $ then $$V\cdot V=\frac{\mu ^2}{r_0^2 v^2}$$ showing that the vector $ V $ has constant length.
\end{proposition}

\begin{proof}

First we notice that $$ \dot{\rho}^{2} = \frac{B^{2}k^{2}}{A^{2}}\sin^{2}(\theta) \quad\hbox{,}\quad \dot{\theta}^{2}\rho^{2} = \frac{k^{2}}{\rho^{2}} = \frac{k^{2}(1-B\cos(\theta))^{2}}{A^{2}} \quad\hbox{, and}\quad $$
$$ \dot{r}\cdot e_{2} = \dot{\rho}\sin(\theta) + \rho\dot{\theta}\cos(\theta) = \frac{-Bk}{A}\sin^{2}(\theta) + \frac{k}{\rho}\cos(\theta) $$

Since $ \dot{r}= \dot{\rho}u_1+\dot{\theta} \rho u_2 $, then 

\begin{eqnarray*}
V\cdot V &=& (\dot{r}+\xi e_2)(\dot{r}+\xi e_2)\\
&=& \dot{r} \cdot \dot{r} + 2\xi\dot{r}\cdot e_{2} + \xi^{2}\\
&=& \dot{\rho}^{2} + \dot{\theta}^{2}\rho^{2} + 2\xi(\dot{\rho}\sin(\theta) + \rho\dot{\theta}\cos(\theta)) + \xi^{2}\\
&=& \frac{B^{2}k^{2}}{A^{2}}\sin^{2}(\theta) + \frac{k^{2}(1-B\cos(\theta))^{2}}{A^{2}} + 2\xi\left(\frac{-Bk}{A}\sin^{2}(\theta) + \frac{k}{\rho}\cos(\theta)\right) + \xi^{2}\\
&\vdots &\\
&=& \frac{B^{2}k^{2}}{A^{2}} + \frac{k^{2}}{A^{2}}     - \frac{2\xi Bk}{A} + \frac{2k\cos(\theta)}{A}\left(\xi - \frac{Bk}{A}\right) + \xi^{2} 
\end{eqnarray*}
In order for $ V \cdot V $ to be constant, $ \xi - \frac{Bk}{A} = 0 $, therefore

$$ \xi = \frac{Bk}{A} = \frac{r_{0}v^{2}-\mu}{r_{0}v} $$ 

then,

$$V\cdot V =\frac{B^{2}k^{2}}{A^{2}} + \frac{k^{2}}{A^{2}}     - \frac{2\xi Bk}{A} + \xi^{2} =\frac{k^{2}}{A^{2}} = \frac{\mu^{2}}{r_{0}^{2}v^{2}} $$

\end{proof}

For all of the planets in our solar system we have that $2\mu-r_0 v^2>0$ and $\mu-r_0 v^2<0$ which together imply that $0<\frac{r_0v^2-\mu}{\mu}<1$.  Therefore the motion of $r$ under the assumption that all other planets are not affecting the motion of the planet under consideration is an ellipse.


\subsection{Parameters in the hodograph theorem in terms of the parameters of the conic}Using the notation and results introduced in this section we will now derive formulas for semi-axes $a$ and $b$, $v = |\dot{r}(0)|$, as well as different formulas for $ V\cdot V $ and $ \xi $ which will be useful when computing the best fitting ellipse.

Recall that $\rho(t)=\frac{A}{1-B \cos(\theta(t))}$ giving us  $$ dm = r_{0} = \rho(0) = \frac{A}{1 + B}, \quad dM = \frac{A}{1 - B} $$ where $dm$ and $dM$ are the distances of perihelion and aphelion respectively.  Also from the above results we have $$ A = \frac{v^2r_0^2}{\mu}, \quad B=\frac{r_0v^2-\mu}{\mu}$$  Lastly we have the relations $$ 2a = dM + dm,\quad a - c = r_{0} = dm, \quad b^{2} = a^{2} - c^{2} $$ (see Figure \ref{fig:ellipse}).

\begin{figure}[h!]
  \includegraphics[keepaspectratio, width=.8\linewidth]{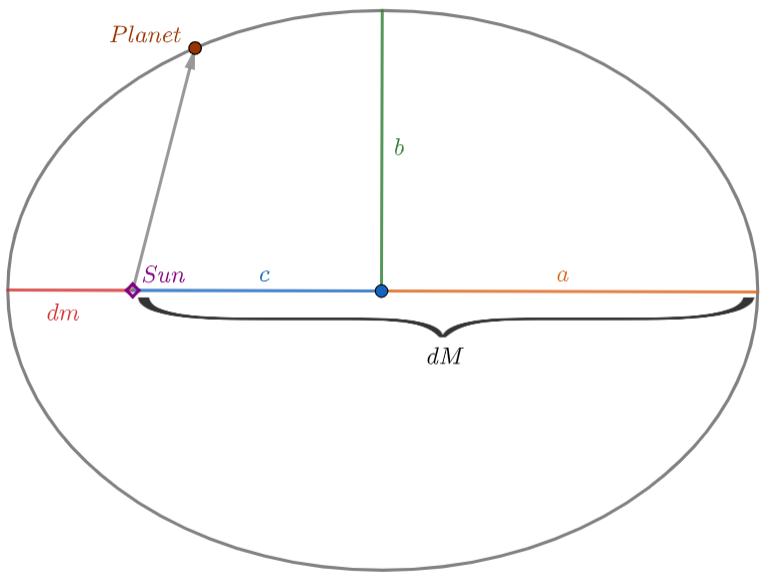} 
  \caption{Illustrating the relations among variables}
  \label{fig:ellipse}
\end{figure}

We begin by finding formulas for the lengths of the semi-axes $a$ and $b$.  By substituting our expressions for $dM$ and $dm$ into the equation $$2a = dM + dm $$ and solving for $a$ gives $$a = \frac{A}{1 - B^{2}} $$  From here we substitute in our expressions for $A$ and $B$, and after simplifying yields $$a=\frac{r_0\mu}{2\mu-r_0v^2}$$  Now since $ c = a - r_{0} $ we have

$$ b^{2} = a^{2} - c^{2} = r_{0}(2a - r_{0}) = r_{0}\left(\frac{2r_0\mu}{2\mu-r_0v^2} - r_{0} \right) = \cdots =  \frac{v^{2}r_{0}^{3}}{2\mu-r_0v^2} $$ 

Therefore the semi-axes of the motion of $r$ are

$$a=\frac{r_0\mu}{2\mu-r_0v^2}\quad\hbox{and}\quad b=\frac{vr_0^{3/2}}{\sqrt{2 \mu - r_0v^2}}$$

From the expression for $a$ we can easily solve for $v$ giving us $$ v = \sqrt{\frac{2\mu}{r_{0}} - \frac{\mu}{a}} $$

Using this expression for $v$ we can now derive alternative formulas for $V\cdot V$ and $\xi$ which will be used in Section \ref{r}.

$$ V\cdot V = \frac{\mu^{2}}{r_{0}^{2}v^{2}} = \frac{\mu^{2}}{r_{0}^{2}\left(\frac{2\mu}{r_{0}} - \frac{\mu}{a}\right)} = \cdots = \frac{a\mu}{r_{0}\left(2a - r_{0}\right)} = \frac{a\mu}{b^{2}} $$ since $ b^{2} = a^{2} - c^{2} = 2ar_{0} - r_{0}^{2} $.

\begin{align*}
\xi &= \frac{r_{0}v^{2} - \mu}{r_{0}v} = \sqrt{\frac{2\mu}{r_{0}} - \frac{\mu}{a}} - \frac{\mu}{r_{0}\sqrt{\frac{2\mu}{r_{0}} - \frac{\mu}{a}}} = \cdots = \sqrt{\frac{2\mu}{r_{0}} - \frac{\mu}{a}} \left(\frac{a-r_{0}}{2a-r_{0}} \right)\\
&= \sqrt{\frac{2\mu}{r_{0}} - \frac{\mu}{a}}\left(\frac{c}{a+c}\right) = \cdots = \sqrt{\frac{\mu c^{2}}{ab^{2}}} = \sqrt{\frac{a\mu}{b^{2}} - \frac{\mu}{a}}  
\end{align*}

Therefore we have that $$\xi = \frac{r_0 v^2-\mu }{r_0 v} = \sqrt{\frac{a\mu}{b^{2}} - \frac{\mu}{a}} \quad \hbox{and}\quad V\cdot V = \frac{\mu ^2}{r_0^2 v^2} = \frac{a\mu}{b^{2}} $$

\section{Best hyperplane}\label{bp}

 Let us assume that we have a collection of points $S=\{p_1,\dots, p_m\}$ in $\mathbb{R}^n$ and that we are interested in finding the best hyperplane $\Pi=\{(x_1,\dots, x_m)\in \mathbb{R}^n: a_1 x_1+\dots+a_n x_n+b=0\}$ that fits these points. A natural way to select this plane is to find the one that minimize the sum of the distances from the points to the plane. If we assume that the vector $a=(a_1,\dots, a_n)$ is a unit vector, then the distance from $p_i$ to the plane $\Pi$ is given by $|p_i\cdot a+b|$ where $u\cdot v$ is the Euclidean dot product. Therefore, the function to minimize is the function  $\sum_{i=1}^n|p_i\cdot a+b|$. For convenience we will be minimizing the function 

$$f(a,b)=\sum_{i=1}^m(p_i\cdot a+b)^2  \quad \hbox{subjected to }  \quad g(a,b)=a\cdot a =1$$

Using Lagrange multiplier, we get the possible minimum happen at points $(a,b,\lambda)\in \mathbb{R}^{n+1}$ such that $\nabla f=\lambda \nabla g$ and $g(a)=1$. A direct computation shows that if $p_i=(x_{i1},x_{i2},\dots, x_{in})$ and we define

$$X_1=(x_{11},\dots, x_{m1})\in \mathbb{R}^m,\, \dots\, X_n=(x_{1n},\dots, x_{mn})\in \mathbb{R}^m\, \hbox{and} \,  u=(1,\dots 1) \in \mathbb{R}^m$$

and the $n\times n$ matrix $M$ with entry $i,j$ given by the dot product $X_i\cdot X_j$, the vector  $v\in \mathbb{R}^n$ with $i^{th}$ entry $X_i\cdot u$, and $L$ the matrix with  entry $i,j$ given by the product $v_i \cdot v_j$then we can rewrite $f$ as follows

\begin{eqnarray*}
f(a,b)&=&\sum_{i=1}^m (p_i\cdot a)^2+2b \sum_{i=1}^m (p_i\cdot a)+mb^2\\
&=&\sum_{i=1}^m (\sum_{j=1}^n x_{ij} a_j)^2+
2b \sum_{i=1}^m \sum_{j=1}^n x_{ij}a_j+mb^2\\
&=& \sum_{i=1}^m \sum_{j=1}^n \sum_{k=1}^n x_{ij}a_jx_{ik} a_k+
2b \sum_{j=1}^n a_j \sum_{i=1}^m  x_{ij} +mb^2\\
&=& a\cdot Ma+mb^2+2b \, a\cdot v \\
\end{eqnarray*}

From the expression above we get that the gradient of $f(a,b)$ is 

$$\nabla f = (2 Ma+2 b v,2 m b+2 a\cdot v)$$

Since $\nabla g=(2a,0)$ then we can write the equations $\nabla f=\lambda \nabla g$ as

$$Ma+b v=\lambda a\quad \hbox{and} \quad m b+a\cdot v=0 $$

Therefore $b=-\frac{1}{m} a\cdot v $ and $Ka=\lambda a$ where $K=M-\frac{1}{m} L$.

\begin{remark}
The arguments above give us the following procedure to compute the best plane that fits the points $S=\{p_1,\dots, p_m\}$ in $\mathbb{R}^n$.

\begin{enumerate}
\item
Compute the vectors $X_1,\dots X_n\in \mathbb{R}^m$, the vector $v\in  \mathbb{R}^n$ and the matrices $M$, $L$, and $K$. Notice that if we view all the vectors as column vectors then $L=vv^T$.
\item
Compute an orthonormal basis $\{w_1,\dots,w_n\}$ that diagonalizes the matrix $K$. Recall that $K$ is a symmetric matrix.
\item
For each $i=1,\dots n$, compute $b_i=-\frac{1}{m} w_i\cdot v$. The equation of the $n$ planes are $\Pi_i=\{x\in \mathbb{R}^n: x\cdot w_i+b_i=0\}$. 
\item
Compute the numbers $f(w_i,b_i)$. If $f(w_k,b_k)\le f(w_i,b_i)$ for all $i$ then $\Pi_k$ is the best plane that fits the collection of points in the set $S$.
\end{enumerate}
\end{remark}

\subsection*{Example 1}

Take the set of $ m = 7 $ points in $ \mathbb{R}^{3} $ to be
$$ S = \lbrace (1, 1, 2), (1, 2.5, 1), (2, 1.5, 1), (3, 2, 1), (3, 3, 2), (4, 3, 2), (1, 0.75, 1) \rbrace $$ 
A direct computation gives us

$$ 
M = 
\begin{pmatrix}
41    && 34.25   && 23\\
34.25 && 32.0625 && 20.75\\
23    && 20.7    && 16
\end{pmatrix} \quad
L = 
\begin{pmatrix}
225    && 206.25  && 150\\
206.25 && 189.063 && 137.5\\
150    && 137.5   && 100
\end{pmatrix}
$$
$$
K = 
\begin{pmatrix}
62/7 & 4.78571 & 11/7\\
4.78571 & 5.05357 & 1.10714\\
11/7 & 1.10714 & 12/7
\end{pmatrix}
$$

where the eigenvalues and eigenvectors of K are

\begin{center}
\begin{tabular}{ l l }
 $ \xi_{1} = 12.449 $ & \quad $ w_{1} = (-0.814171, -0.553232, -0.176239) $ \\ 
 $ \xi_{2} = 1.80811 $ & \quad $ w_{2} = (0.571409, -0.817313, -0.0741071) $ \\  
 $ \xi_{3} = 1.36767 $ & \quad $ w_{3} = (-0.103044, -0.161041, 0.981554) $   
\end{tabular}
\end{center}

Computing our constant terms $ b_{i} $ we get three candidates for the plane of best fit

\begin{flalign*}
\quad\quad\quad& \Pi_{1}: 3.08313 - 0.814171 x - 0.553232 y - 0.176239 z = 0 &\\
& \Pi_{2}: 0.486855 + 0.571409 x - 0.817313 y - 0.0741071 z = 0 &\\
& \Pi_{3}: -0.865081 - 0.103044 x - 0.161041 y + 0.981554 z = 0 &
\end{flalign*}

For these values of $w_i$  and $b_i$, $ f(w_{1}, b_{1}) \approx 12.4492 $, $ f(w_{2}, b_{2}) \approx 1.80811 $ and $ f(w_{3}, b_{3}) \approx 1.36767 $.  Therefore our plane of best fit is $ \Pi_{3} $.

\begin{figure}[ht!]
  \centering
  \begin{subfigure}[b]{1\linewidth}
    \includegraphics[width=\linewidth]{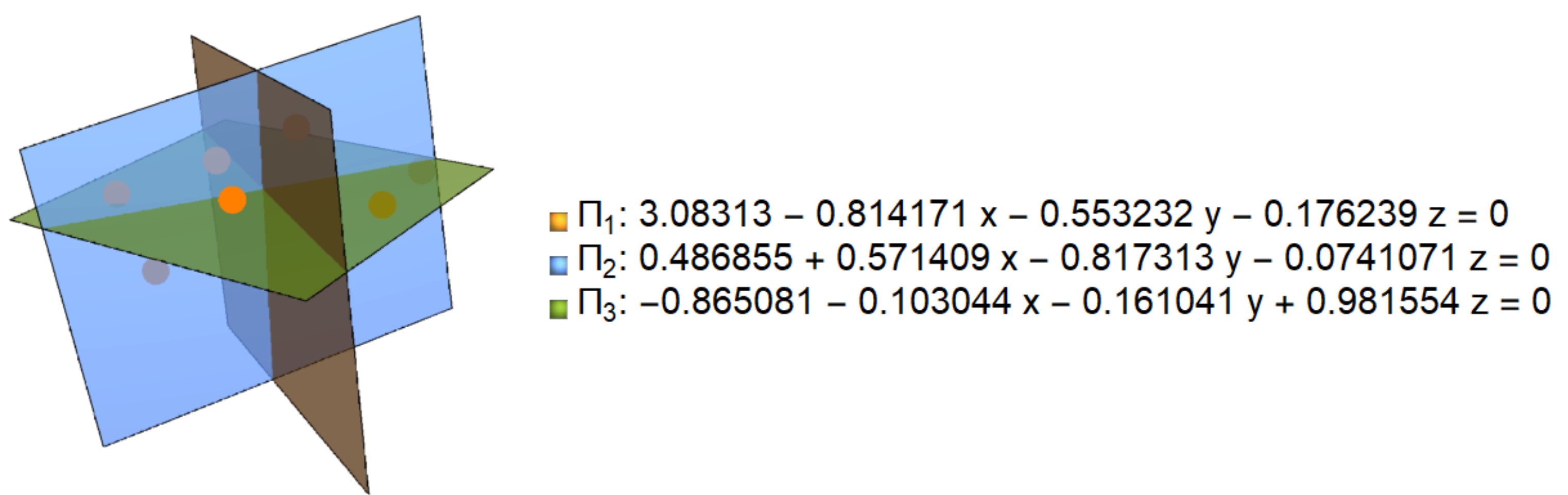}
  \end{subfigure}
  \caption{$ \Pi_{1}, \Pi_{2},\Pi_{3} $ divide $ \mathbb{R}^{3} $ into eight octants.  In this case $\Pi_3$ is the best plane.}
\end{figure}

\section{The best ellipse}\label{be}

We will be considering two approaches for finding the best fitting ellipse given a set of points in 
$S=\{(x_1,y_1),\dots ,(x_m,y_m)\}$ in $ \mathbb{R}^{2} $. The first one uses only the set of points under the small assumption that the best ellipse does not contain the origin, and the second assumes that we know the eccentricity of the ellipse and that one of  the semi-axes is horizontal. For both procedures we write the desired equation of the ellipse in the form $g(x,y,A,B,\dots)=0$ where $A,B,\dots$ are coefficients of a polynomial of order 2 in the variables $(x,y)$.  In the case where all the points of $ S $ are in an ellipse and the values for $A,B,\dots$ are such that $g(x,y,A,B,\dots) = 0$ describes the ellipse for the initial points of $ S $, then the function 
$$f(A,B,\dots)=\sum_{i=1}^m g(x_i,y_i,A,B\dots)^2$$
is zero.  In the case that the set of points is not an ellipse, for any $A,B\dots$ the function $f$ is greater than zero.  So $f$ can be taken as a distance that describes how far the ellipse with equation $g(x,y,A,B,\dots)=0$ is from the ellipse that best fits the data set of points.  With this observation in mind, we find our best ellipse my minimizing the function $f(A,B,\dots)$.  

\subsection*{Approach 1} 
In this case we will use the fact that any ellipse that does not contain the origin can be described with the equation $g=0$ where $ g(x,y,A,B,C,D,E)=Ax^{2} + Bxy + Cy^{2} + Dx + Ey -1 $, in this case,
$$ f(A, B, C, D, E) = \sum_{i=1}^{m}(Ax_{i}^{2} + Bx_{i}y_{i} + Cy_{i}^{2} + Dx_{i} + Ey_{i} - 1)^{2} $$
 
Expanding $ f $ and taking partial derivatives yields  

\begin{align*}
\frac{\partial f}{\partial A} &= 2A(X_{4}Y_{0}) + 2B(X_{3}Y_{1}) + 2C(X_{2}Y_{2}) + 2D(X_{3}Y_{0}) + 2E(X_{2}Y_{1}) - 2(X_{2}Y_{0})\\
\frac{\partial f}{\partial B} &= 2A(X_{3}Y_{1}) + 2B(X_{2}Y_{2}) + 2C(X_{1}Y_{3}) + 2D(X_{2}Y_{1}) + 2E(X_{1}Y_{2}) + -2(X_{1}Y_{1})\\
\frac{\partial f}{\partial C} &= 2A(X_{2}Y_{2}) + 2B(X_{1}Y_{3}) + 2C(X_{0}Y_{4}) + 2D(X_{1}Y_{2}) + 2E(X_{0}Y_{3}) - 2(X_{0}Y_{2})\\
\frac{\partial f}{\partial D} &= 2A(X_{3}Y_{0}) + 2B(X_{2}Y_{1}) + 2C(X_{1}Y_{2}) + 2D(X_{2}Y_{0}) + 2E(X_{1}Y_{1}) - 2(X_{1}Y_{0})\\
\frac{\partial f}{\partial E} &= 2A(X_{2}Y_{1}) + 2B(X_{1}Y_{2}) + 2C(X_{0}Y_{3}) + 2D(X_{1}Y_{1}) + 2E(X_{0}Y_{2}) - 2(X_{0}Y_{1})
\end{align*}

such that $ X_{j}Y_{k} = \sum_{i=1}^{m}x_{i}^{j}y_{i}^{k} $.

Setting each equation equal to zero yields the matrix equation $ Ax = b $ where

$$ A = \begin{pmatrix}
X_{4}Y_{0} & X_{3}Y_{1} & X_{2}Y_{2} & X_{3}Y_{0} & X_{2}Y_{1}\\
X_{3}Y_{1} & X_{2}Y_{2} & X_{1}Y_{3} & X_{2}Y_{1} & X_{1}Y_{2}\\
X_{2}Y_{2} & X_{1}Y_{3} & X_{0}Y_{4} & X_{1}Y_{2} & X_{0}Y_{3}\\
X_{3}Y_{0} & X_{2}Y_{1} & X_{1}Y_{2} & X_{2}Y_{0} & X_{1}Y_{1}\\
X_{2}Y_{1} & X_{1}Y_{2} & X_{0}Y_{3} & X_{1}Y_{1} & X_{0}Y_{2}
\end{pmatrix} \quad x = \begin{pmatrix}
A\\B\\C\\D\\E
\end{pmatrix} \quad b = \begin{pmatrix}
X_{2}Y_{0}\\X_{1}Y_{1}\\X_{0}Y_{2}\\X_{1}Y_{0}\\X_{0}Y_{1}
\end{pmatrix} $$

Notice that $ A $ is a symmetric matrix.  Assuming the $ Det(A) \neq 0 $ gives us a unique solution $ x = A^{-1}b $. 

\subsection*{Approach 2}

Assume we have a set of points in $ \mathbb{R}^{2} $ that approximate an axis aligned ellipse with eccentricity $ e $ that does not contain the origin.  Then a direct computation shows that we can look for an ellipse of the form $$ (1-e^{2})Bx^{2} + By^{2} + Cx + Dy - 1 = 0 $$

Using the same notation and following the same procedure as in approach 1 we have that $ Ax  = b $ where

$$ A = \begin{pmatrix}
X_{4}Y_{0}(1-e^{2})^{2} + 2X_{2}Y_{2}(1-e^{2}) + X_{0}Y_{4} & X_{3}Y_{0}(1-e^{2}) + X_{1}Y_{2} & X_{2}Y_{1}(1-e^{2}) + X_{0}Y_{3}\\
X_{3}Y_{0}(1-e^{2}) + X_{1}Y_{2} & X_{2}Y_{0} & X_{1}Y_{1}\\
X_{2}Y_{1}(1-e^{2}) + X_{0}Y_{3} & X_{1}Y_{1} & X_{0}Y_{2} \end{pmatrix} $$

$$ x = \begin{pmatrix}
B\\C\\D
\end{pmatrix} \quad b = \begin{pmatrix}
X_{2}Y_{0}(1-e^{2}) + X_{0}Y_{2}\\
X_{1}Y_{0}\\
X_{0}Y_{1}
\end{pmatrix} $$

Notice that $ A $ is a symmetric matrix.  Assuming the $ Det(A) \neq 0 $ gives us a unique solution $ x = A^{-1}b $. 

\subsection*{Example of Approach 1}

Take the set of $ m = 7 $ points in $ \mathbb{R}^{2} $ to be
$$ S = \lbrace (2, 1), (2, 4), (3, 1), (3, 6), (4, 2), (5, 4), (5, 6) \rbrace $$ 
A direct computation gives us

$$ A = 
\begin{pmatrix}
1700 && 1607 && 1765 && 384 && 365\\ 
1607 && 1765 && 2213 && 365 && 421\\ 
1765 && 2213 && 3122 && 421 && 570\\
384  && 365  && 421  && 92  && 89\\
365  && 421  && 570  && 89  && 110
\end{pmatrix} \quad b = \begin{pmatrix}
92\\89\\110\\24\\24
\end{pmatrix}
$$

where $ Det(A) = 57,566,592 $ and therefore has a unique inverse.  Therefore we have as a solution 

$$ x = A^{-1}b \approx \begin{pmatrix}
-0.15008\\0.0885364\\-0.0512472\\0.684014\\0.0894444
\end{pmatrix} $$

Thus the equation of our best fitting ellipse is given by 

$$ - 0.15008 x^2 + 0.0885364 xy - 0.0512472 y^2 + 0.684014 x + 0.0894444 y - 1 =0 $$

\begin{figure}[h!]
  \includegraphics[width=.4\linewidth]{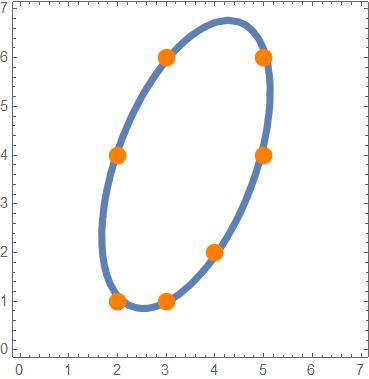}
  \caption{}
\end{figure}

\subsection*{Example of Approach 2}

While in the celestial-mechanic application that we will be explaining later in this paper we will use other methods to compute the eccentricity and angle the major axis makes with the positive x-axis, to illustrate Approach 2 we use our result from Approach 1 as a starting point.  Using the equation we found we directly compute the eccentricity and angle the major axis of our ellipse makes with the positive x-axis giving us $$ e \approx 0.891353, \quad \theta \approx 1.205545 $$

For each point $ p_{i} \in S $ we perform a change of basis by computing $ (p_{i} \cdot (\cos\theta, \sin\theta), p_{i} \cdot (-\sin\theta, \cos\theta)) $ giving us $$ S' = \{(1.6484, -1.51088), (4.4505, -0.439331), (2.00559, -2.44492), $$ $$ 
(6.67576, -0.658996), (3.2968, -3.02177), (5.52206, -3.24143),
(7.39012, -2.52706)\} $$

A direct computation shows

$$ A = 
\begin{pmatrix}
885.103 && 361.556 && -162.33\\
361.556 && 167.088 && -60.2855\\
-162.33 && -60.2855 && 34.9117
\end{pmatrix} \quad b = \begin{pmatrix}
69.2468\\30.9892\\-13.8444
\end{pmatrix}
$$

where $ Det(A) = 56113.6 $ and therefore has a unique inverse.  Thus the solution is

$$ x = A^{-1}b \approx \begin{pmatrix}
-0.167008\\0.327863\\-0.606944
\end{pmatrix} $$

The equation of our best fitting ellipse is given by 

$$ - 0.0343186 x^2 + 0.327863 x - 0.167008 y^2 - 0.606944 y  - 1 = 0 $$

\begin{figure}[h!]
  \includegraphics[width=.6\linewidth]{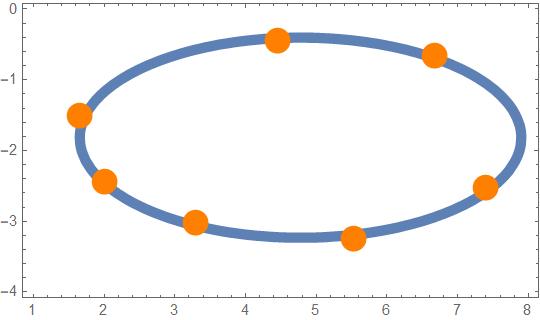}
  \caption{}
\end{figure}

Note that if we wish to have the equation of this ellipse model our original (non-axis aligned) set of points, we simply use the angle $\theta$ found above and replace $ x \rightarrow x\cos\theta + y\sin\theta $ and $ y \rightarrow -x\sin\theta + y\cos\theta $ giving us the equation found above in the example of approach 1 $$ - 0.15008 x^2 + 0.0885364 xy - 0.0512472 y^2 + 0.684014 x + 0.0894444 y - 1 =0 $$

\subsection{Best ellipse for a collection of points in $\mathbb{R}^3$} We can combine the procedure of finding the plane of best fit and the ellipse of best fit to find the best ellipse for a given collection of points in $\mathbb{R}^3$.  Given a collection of points $S=\{p_1,\dots, p_m\}$ in $\mathbb{R}^3$ we find the plane of best fit $\Pi$. Letting $n$ be the unit normal vector to this plane and $E_1$, $E_2$ two orthonormal vectors perpendicular to $n$ we now consider the set of points in the space $S^\prime=\{p^\prime_1,\dots p^\prime_m\}$ where $p^\prime_i=(p_i\cdot E_1,p_i\cdot E_1)$. Essentially the points $p^\prime_i$ are the coordinates of the projection of the points $p_i$ on a plane $\Pi^\prime$ that contains the origin and is parallel to the plane $\Pi$. We can now use the procedure to find the best ellipse with the planar points $p^\prime_i$, and after, we use the relation between $\Pi$, $\Pi^\prime$ to find the ellipse in $\mathbb{R}^3$ that best fits the initial collections of points. We now show an example of this procedure.

Take a semi-random set of points in $\mathbb{R}^{3}$. Take the set of $m = 10$ points in $\mathbb{R}^{3}$ to be 
\begin{align*}
S &= \lbrace(2, 3.37, 3.45), (0.3, 2.07, 2.26), (0.61,1.29, -0.27), (1.24, 1.17, -0.55), (2.99, 2.04, 0.95),\\
  &\quad \quad (3.48, 3, 3.56), (2.58, 3.8, 2.89), (0.14, 1.49, 1.48), (0.06, 0.56, 0.96), (0.64, 0.57, 0.62) \rbrace
\end{align*}

Using the procedure for the plane of best fit we find $$ \Pi: 0.332541 x - 0.840721 y + 0.427323 z + 0.504806 = 0 $$ with $n = (0.332541, - 0.840721, 0.427323)$ and $ d = 0.504806$ 

Now letting $\Pi'$ be the parallel plane through the origin, we compute basis vectors
\begin{align*}
E1 &= (0.943089, 0.296445, -0.150678)\\
E2 &= (1.02198 \cdot 10^{-17}, 0.453111, 0.891454)
\end{align*}

Projecting our set of points onto the basis vectors and using our best ellipse procedure gives us $$ \alpha (t) = (2.05437, 2.49305) + 2.34605\cos t (\cos u,\sin u) + 1.5626\sin t (-\sin u, \cos u) $$ where $u = -2.02649$ and $t \in [0, 2\pi)$.  

Note that $\alpha (t) \in \mathbb{R}^{2}$ so our final step is to put it back onto the plane of best fit $\Pi$ through the original set of points in $\mathbb{R}^{3}$.

Taking  
\begin{align*}
p_{0} &= (-0.167869, 0.424401, -0.215715)\\
\eta_{1} &= 2.05437 - 1.03246 \cos t + 1.40315 \sin t\\
\eta_{2} &= 2.49305 - 2.10665 \cos t - 0.687677 \sin t
\end{align*}

where $p_{0} = -dn$ and $\eta_{1}, \eta_{2}$ are the first and second entries of $\alpha (t)$ gives us $$ \phi (t) = p_{0} + \eta_{1}E_{1} + \eta_{2}E_{2} \in \mathbb{R}^{3} \quad \hbox{for} \quad t \in [0, 2\pi) $$ as our ellipse of best fit through the set of points $p_{i}$ projected onto the plane $\Pi$.  See Figure \ref{fig:ellipsePlane}.

\begin{figure}[h!]
  \includegraphics[width=.55\linewidth]{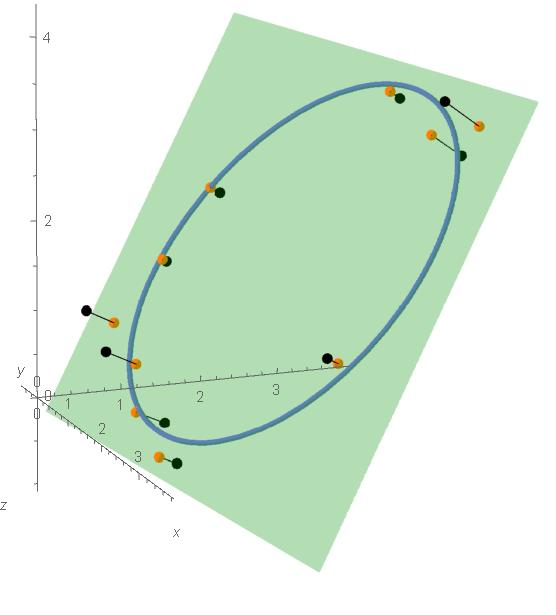}
  \caption{Image showing our plane $\Pi$, ellipse $\phi$, and points $p_{i}$ with their projections}
  \label{fig:ellipsePlane}
\end{figure}

\section{Procedure for getting the ellipse}\label{p}

First we explain the full procedure for finding the parametrization for Earth.  After that we list the results for Mercury, Venus, Earth, Mars, and Jupiter all together.

Since we are finding equations to model relations, we must define our coordinate system that is being used.  We use the settings given to us in the Horizons Web Interface.  The origin of our coordinate system is taken to be the Solar System Barycenter.  This is the center of mass of our solar system.  

The reference plane used is the ecliptic and mean equinox of reference epoch, and the reference system is ICRF/J2000.0 Documentation regarding reference frames and coordinate systems can be found on the Horizon documentation page \url{https://ssd.jpl.nasa.gov/?horizons_doc#frames}

\subsection{Downloading and importing data}

Navigating to the NASA HORIZONS Web-Interface at \url{https://ssd.jpl.nasa.gov/horizons.cgi} we download the data for the Earth and the Sun using the following settings  

\begin{itemize}
\item Ephemeris Type: \textbf{Vectors}
\item Target Body: \textbf{[Your Choice]}
\item Coordinate Origin: \textbf{Solar System Barycenter}
\item Time Span: Start=\textbf{2020-01-01}, Stop=\textbf{2020-12-31}, Step=\textbf{1 d}
\item Table Settings: quantities code=\textbf{2}; output units=\textbf{KM-S}; CSV format=\textbf{YES}
\item Display/Output: \textbf{download/save} (plain text file)
\end{itemize}

This provides us with two files (one for the sun, and the other for earth) containing the date of each observation, along with the $x$, $y$, $z$ position coordinates, and $vx$, $vy$, $vz$ velocity coordinates respectively.

We now consider the relative positions and velocities of the earth with respect to the sun taking $i = 1, 2, \ldots, 366$ giving us one full rotation around the sun.
\begin{align*}
p_{i} &= (x, y, z)_{Earth} - (x, y, z)_{Sun}\\
v_{i} &= (vx, vy, vz)_{Earth} - (vx, vy, vz)_{Sun}
\end{align*}

For what follows, for each planet we pick dates that correspond with a single period/year, or one full rotation around the sun.

\subsection{Plane of best fit and projecting points into $ \mathbb{R}^{2} $}

With a table in Mathematica now containing our relative position and velocity vectors, we are able to use the procedure outlined in Section \ref{bp} to compute the plane of best fit $ \Pi: ax + by + cz + d = 0 $ with normal vector $ n = (a, b, c) $.

Using our plane of best fit we map our points from $ \mathbb{R}^{3} $ into $ \mathbb{R}^{2} $ using the following procedure

\begin{enumerate}
\item Using the normal vector $ n = (a, b, c) $ from our plane of best fit, we consider the plane $ \Pi ' $ with normal vector $ n $ such that $ \Pi ' $ contains the origin.  This plane is given by the equation $ ax + by + cz = 0 $
\item Take the orthonormal basis for $ \Pi ' $ given by \begin{align*}
E_{1} &= \frac{e_1 - (e_1 \cdot n) n}{\lVert e_1 - (e_1 \cdot n) n\rVert}\\
E_{2} &= \frac{n \times E_{1}}{\lVert n \times E_{1} \rVert}
\end{align*} 
Notice that $ E_{1} $ is the unit vector of the projection of $ e_1 = (1, 0, 0) $ into $ \Pi ' $ and $ Span\{E_{1}, E{2}\} = \Pi ' $.  
\item We now project each point onto our basis vectors $ E_{1} $ and $ E_{2} $ giving us a new set of coordinates $$ p_{i}' = (p_{i}\cdot E_{1}, p_{i}\cdot E_{2}) \quad \hbox{and}\quad v_{i}' = (v_{i}\cdot E_{1}, v_{i}\cdot E_{2}) $$
\end{enumerate}

\subsection{Ellipse of best fit and eccentricity}

Let us temporarily assume that our planet moves according to the solution of the two-body problem in a perfect ellipse with major semi axis $\tilde{a}$ and minor semi-axis $\tilde{b}$ parallel to the direction given by the unit vector $e_2$. According to our results from Section \ref{bp} we know that if $\tilde{\xi}=\sqrt{\frac{\tilde{a}\mu}{\tilde{b}^{2}} - \frac{\mu}{\tilde{a}}}$ and $\tilde{V}_i=v_{i}' -\tilde{\xi} e_2$, then $\tilde{d_i} = \tilde{V}_i\cdot \tilde{V}_i=\frac{\tilde{a}\mu}{\tilde{b}^{2}}$ for all $i$, in particular

$$  \sum_{i=1}^{m}(\tilde{V}_i\cdot \tilde{V}_i - \frac{\tilde{a}\mu}{\tilde{b}^{2}})^{2} = 0 $$  

and therefore, in this hypothetical case that the motion moves according to the solution of the two-body problem, we can find three real numbers $\tilde{d}$, $\tilde{\xi}$ and $\tilde{u}$ such that the function

$$ f(\tilde{\xi},\tilde{d},\tilde{u})= \sum_{i=1}^{m}((v_{i}' -\tilde{\xi} (\cos \tilde{u},sin \tilde{u}))\cdot (v_{i}' -\tilde{\xi} (\cos \tilde{u},sin \tilde{u})) - \tilde{d})^2 = 0 $$  

for all $v_i$.  Remember that the ${v_i}'$ come from the data taken from Nasa, and in the hypothetical case we are assuming they will satisfy the relations above taken from the hodograph theorem. With this observation in mind, we use the function $f$, which is a nonnegative  function to measure how far the real motion of the planet is from the motion of the theoretical ellipse.  We will minimize the function $f$ using Mathematica. Since the function $f$ may have several local minima, we need to search for a minimum of $f$ near the the point $(\tilde{\xi},\tilde{d},\tilde{u})\approx(\xi_r,d_r,0)$ with
 
 $$\xi_r=\sqrt{\frac{a_r\mu}{b_r^{2}} - \frac{\mu}{a_r}}\quad\hbox{and}\quad d_r=\frac{a_r\mu}{b_r^2}$$
 
 where,
 
 $$a_r=\frac{dM+dm}{2}\hbox{ and } b_r=\sqrt{a_r^2-(a_r-dm)^2}=\sqrt{dm\cdot dM}$$ 
 
 with  $dM$ and $dm$ taken to be the maximum and minimum magnitudes of our vectors $p_{i}$.  Note that the subscript $_r$ notation used above represent our first rough estimates. 
So with our starting estimates in hand, we will use a gradient descent algorithm in Mathematica to find $ \xi^* $, $ d^* $ , and $u^*$ that minimize the function $f(\tilde{\xi},\tilde{d},\tilde{u})$.
Once we have our optimal $\xi^*$, $u^*$, and $d^*$ we can solve the system 
\begin{align*}
\xi^* &= \sqrt{\frac{a^*\mu}{{b^*}^{2}} - \frac{\mu}{a^*}}\\
d^*   &= \frac{a^*\mu}{{b^*}^{2}}
\end{align*}
for $a^*$ and $b^*$.  At this point the eccentricity is easily computed $ e = \frac{\sqrt{{a^*}^{2} - {b^*}^{2}}}{a^*} $.  Taking the angle $u^*$ given by the gradient descent algorithm we perform one final change of coordinates given by 

$$ p_{i}'' = (p_{i}'\cdot (\cos u^*, \sin u^*), p_{i}'\cdot (-\sin u^*, \cos u*)) $$  

Recall that $(\cos u^*, \sin u^*)$ is perpendicular to the major axis, so that we end up with a set of coordinates that represents an axis-aligned ellipse (no rotation with respect to the positive x-axis).  We can now use the ellipse of best fit Approach 2 explained in Section \ref{be} to get the equation of the ellipse that represents the orbit of our planet with the specified eccentricity.

We now have an equation of the form $$ Ax^{2} +Bx + Cy^{2} + Dy - 1 = 0 $$ for some constants $A, B, C, D \in \mathbb{R} $.  Note that if we wish to have the equation of this ellipse model our original (non-axis aligned) set of points, we simply use the angle $u$ found above and replace $ x \rightarrow x\cos u^* + y\sin u^* $ and $ y \rightarrow -x\sin u^* + y\cos u^* $ 

We will now compute the parametrization  of the original, rotated ellipse.  First we complete the square of $ Ax^{2} +Bx + Cy^{2} + Dy - 1 = 0 $ giving us
$$ \frac{\left( x + \frac{B}{2A} \right)^{2}}{\left( \frac{AD^{2} + CB^{2} + 4AC}{4A^{2}C} \right)} + \frac{\left( y + \frac{D}{2C} \right)^{2}}{\left( \frac{AD^{2} + CB^{2} + 4AC}{4AC^{2}} \right)} = 1 $$

Therefore 
\begin{align*}
h &= (\frac{-B}{2A}, \frac{-D}{2C}) \cdot (\cos u, -\sin u)\\
k &= (\frac{-B}{2A}, \frac{-D}{2C}) \cdot (\sin u, \cos u)\\
a^* &= \sqrt{\frac{AD^{2} + CB^{2} + 4AC}{4A^{2}C}}\\
b^* &= \sqrt{\frac{AD^{2} + CB^{2} + 4AC}{4AC^{2}}}
\end{align*}

and our parametrization is $$ \alpha(t) = (h, k) + a^*\cos t (\cos u, \sin u) + b^* \sin t (-\sin u, \cos u) $$ for $ t \in [0, 2\pi) $.

We now take this parametrization $\alpha(t)$ in $\mathbb{R}^{2}$ and put it back into $\mathbb{R}^{3}$ onto its original position in the plane of best fit.  Recall that our plane of best fit is $ \Pi: ax + by + cz + d = 0 $ with normal vector $ n = (a, b, c) $.  We would like to translate our points from $ \Pi': ax + by + cz = 0 $ back to $ \Pi $.  To do this we need to find the point on $\Pi$ that corresponds to the origin of $\Pi'$.  In other words, where the vector $n$ intersects $\Pi$.  

Letting $p_{0} = tn = t(a, b, c)$ for some $t \in \mathbb{R}$ and substituting $p_{0}$ into the equation for $\Pi$ gives us $ (a^{2} + b^{2} + c^{2})t + d = 0 $ and since $ n $ is a unit vector $t = -d$.  Thus $p_{0} = -dn = -d(a, b, c)$.

Therefore the parametrization of the best fitting ellipse through our original set of points $p_{i} \in \mathbb{R}^{3}$ is given by $$ \phi(t) = p_{0} + \eta_{1}E_{1} + \eta_{2}E_{2} $$ where $E_{1}, E_{2}$ are defined as above and 
\begin{align*}
\eta_{1} &= h + a^*\cos t \cos u - b^*\sin t\sin u\\
\eta_{2} &= k + a^*\cos t \sin u + b^*\sin t\cos u
\end{align*}

\section{Results}\label{r}
The number of data points used for each planet is the number of days in one period (orbit around the sun).  The date range of the data used is given next to the planet name.  All results are in kilometers.

\begin{tabular}{ | l || c | }
\multicolumn{2}{l}{{\large \textbf{Mercury}} (2021.01.01 - 2021.03.29)} \\
\hline
 \textbf{Semi-major axis} & $5.79523\cdot 10^7$ km \\ 
 \textbf{Semi-minor axis} & $5.67138\cdot 10^7$ km \\  
 \textbf{Eccentricity} & $0.205637$    \\
 \textbf{Parametrization} & $ \phi (t) = p_{0} + \eta_{1}E_{1} + \eta_{2}E_{2} \in \mathbb{R}^{3} \quad \hbox{for} \quad t \in [0, 2\pi) $ \\
 \hline 
 \multicolumn{2}{| l |}{$E1 = (0.995847, 0.00741525, -0.0907438)$}\\ 
 \multicolumn{2}{| l |}{$E2 = (1.17303\cdot 10^{-18}, 0.996678, 0.0814449)$}\\ 
 \multicolumn{2}{| l |}{$p_{0} = (-0.552638, 0.492306, -6.02457)$}\\ 
 \multicolumn{2}{| l |}{$\eta_{1} = -2.64667\cdot 10^6 - 5.65296\cdot 10^7 \cos t - 1.24896\cdot 10^7 \sin t$}\\ 
 \multicolumn{2}{| l |}{$\eta_{2} = -1.1722\cdot 10^7 + 1.27623\cdot 10^7 \cos t - 5.53214\cdot 10^7 \sin t$} \\
 \hline
\end{tabular}

\begin{tabular}{ | l || c | }
\multicolumn{2}{l}{{\large \textbf{Venus}} (2020.01.01 - 2020.08.12)} \\
\hline
 \textbf{Semi-major axis} & $1.08209\cdot 10^8$ km \\ 
 \textbf{Semi-minor axis} & $1.08207\cdot 10^8$ km \\  
 \textbf{Eccentricity} & $0.00675998$    \\
 \textbf{Parametrization} & $ \phi (t) = p_{0} + \eta_{1}E_{1} + \eta_{2}E_{2} \in \mathbb{R}^{3} \quad \hbox{for} \quad t \in [0, 2\pi) $ \\
 \hline 
 \multicolumn{2}{| l |}{$E1 = (0.998339, 0.000790377, -0.0576002)$}\\ 
 \multicolumn{2}{| l |}{$E2 = (-3.64763\cdot 10^{-20}, 0.999906, 0.0137205)$}\\ 
 \multicolumn{2}{| l |}{$p_{0} = (-1.55863, 0.370618, -27.0095)$}\\ 
 \multicolumn{2}{| l |}{$\eta_{1} = 484657. + 8.10391\cdot 10^7 \cos t - 7.17057\cdot 10^7 \sin t$}\\ 
 \multicolumn{2}{| l |}{$\eta_{2} = -548351. + 7.17073\cdot 10^7 \cos t + 8.10373\cdot 10^7 \sin t$} \\
 \hline
\end{tabular}

\begin{tabular}{ | l || c | }
\multicolumn{2}{l}{{\large \textbf{Earth}} (2020.01.01 - 2020.12.31)} \\
\hline
 \textbf{Semi-major axis} & $1.49598\cdot 10^8$ km \\ 
 \textbf{Semi-minor axis} & $1.49577\cdot 10^8$ km \\  
 \textbf{Eccentricity} & $0.0167143$    \\
 \textbf{Parametrization} & $ \phi (t) = p_{0} + \eta_{1}E_{1} + \eta_{2}E_{2} \in \mathbb{R}^{3} \quad \hbox{for} \quad t \in [0, 2\pi) $ \\
 \hline 
 \multicolumn{2}{| l |}{$E1 = (1., -1.29884\cdot 10^{-10}, -2.80132\cdot 10^{-6})$}\\ 
 \multicolumn{2}{| l |}{$E2 = (-1.93088\cdot 10^{-27}, 1., -0.0000463654)$}\\ 
 \multicolumn{2}{| l |}{$p_{0} = (0.0000523882, 0.000867092, 18.7013)$}\\ 
 \multicolumn{2}{| l |}{$\eta_{1} = 562946. + 1.45758\cdot 10^8 \cos t - 3.36704\cdot 10^7 \sin t$}\\ 
 \multicolumn{2}{| l |}{$\eta_{2} = -2.43893\cdot 10^6 + 3.36751\cdot 10^7 \cos t + 1.45738\cdot 10^8 \sin t$} \\
 \hline
\end{tabular}

\begin{tabular}{ | l || c | }
\multicolumn{2}{l}{{\large \textbf{Mars}} (2019.01.01 - 2020.11.17)} \\
\hline
 \textbf{Semi-major axis} & $2.27948\cdot 10^8$ km \\ 
 \textbf{Semi-minor axis} & $2.26951\cdot 10^8$ km \\  
 \textbf{Eccentricity} & $0.0934294$    \\
 \textbf{Parametrization} & $ \phi (t) = p_{0} + \eta_{1}E_{1} + \eta_{2}E_{2} \in \mathbb{R}^{3} \quad \hbox{for} \quad t \in [0, 2\pi) $ \\
 \hline 
 \multicolumn{2}{| l |}{$E1 = (0.999699, 0.0005137, -0.0245161)$}\\ 
 \multicolumn{2}{| l |}{$E2 = (1.17962\cdot 10^{-19}, 0.999781, 0.020949)$}\\ 
 \multicolumn{2}{| l |}{$p_{0} = (5.00288, -4.27273, 203.914)$}\\ 
 \multicolumn{2}{| l |}{$\eta_{1} = -1.95162\cdot 10^7 + 9.21814\cdot 10^7 \cos t - 2.07565\cdot 10^8 \sin t$}\\ 
 \multicolumn{2}{| l |}{$\eta_{2} = 8.62918\cdot 10^6 + 2.08477\cdot 10^8 \cos t + 9.17782\cdot 10^7 \sin t$} \\
 \hline
\end{tabular}

\begin{tabular}{ | l || c | }
\multicolumn{2}{l}{{\large \textbf{Jupiter}} (2009.01.01 - 2020.11.09)} \\
\hline
 \textbf{Semi-major axis} & $7.7827\cdot 10^8$ km \\ 
 \textbf{Semi-minor axis} & $7.77341\cdot 10^8$ km \\  
 \textbf{Eccentricity} & $0.04884251$    \\
 \textbf{Parametrization} & $ \phi (t) = p_{0} + \eta_{1}E_{1} + \eta_{2}E_{2} \in \mathbb{R}^{3} \quad \hbox{for} \quad t \in [0, 2\pi) $ \\
 \hline 
 \multicolumn{2}{| l |}{$E1 = (0.99975, -0.0000929091, -0.0223711)$}\\ 
 \multicolumn{2}{| l |}{$E2 = (1.0625\cdot 10^{-20}, 0.999991, -0.00415306)$}\\ 
 \multicolumn{2}{| l |}{$p_{0} = (-8.40464, -1.55987, -375.592)$}\\ 
 \multicolumn{2}{| l |}{$\eta_{1} = -3.68265\cdot 10^7 - 1.92357\cdot 10^8 \cos t - 7.53224\cdot 10^8 \sin t$}\\ 
 \multicolumn{2}{| l |}{$\eta_{2} = -9.41837\cdot 10^6 + 7.54124\cdot 10^8 \cos t - 1.92127\cdot 10^8 \sin t$} \\
 \hline
\end{tabular}

\begin{figure}[h!]
  \includegraphics[width=.8\linewidth]{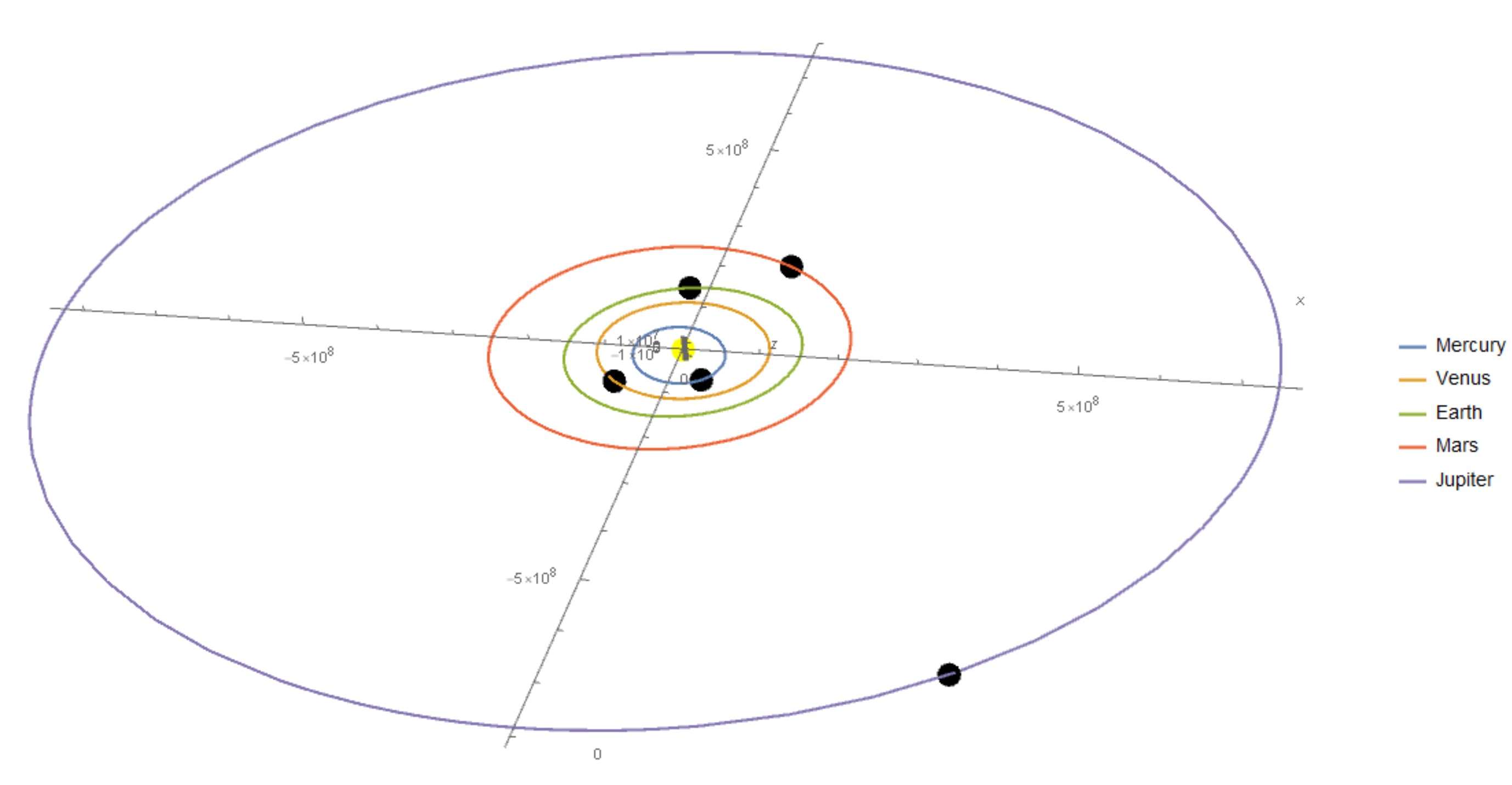}
  \caption{Angled view: Positions on January 1, 2021}
\end{figure}

\begin{figure}[h!]
  \includegraphics[width=.8\linewidth]{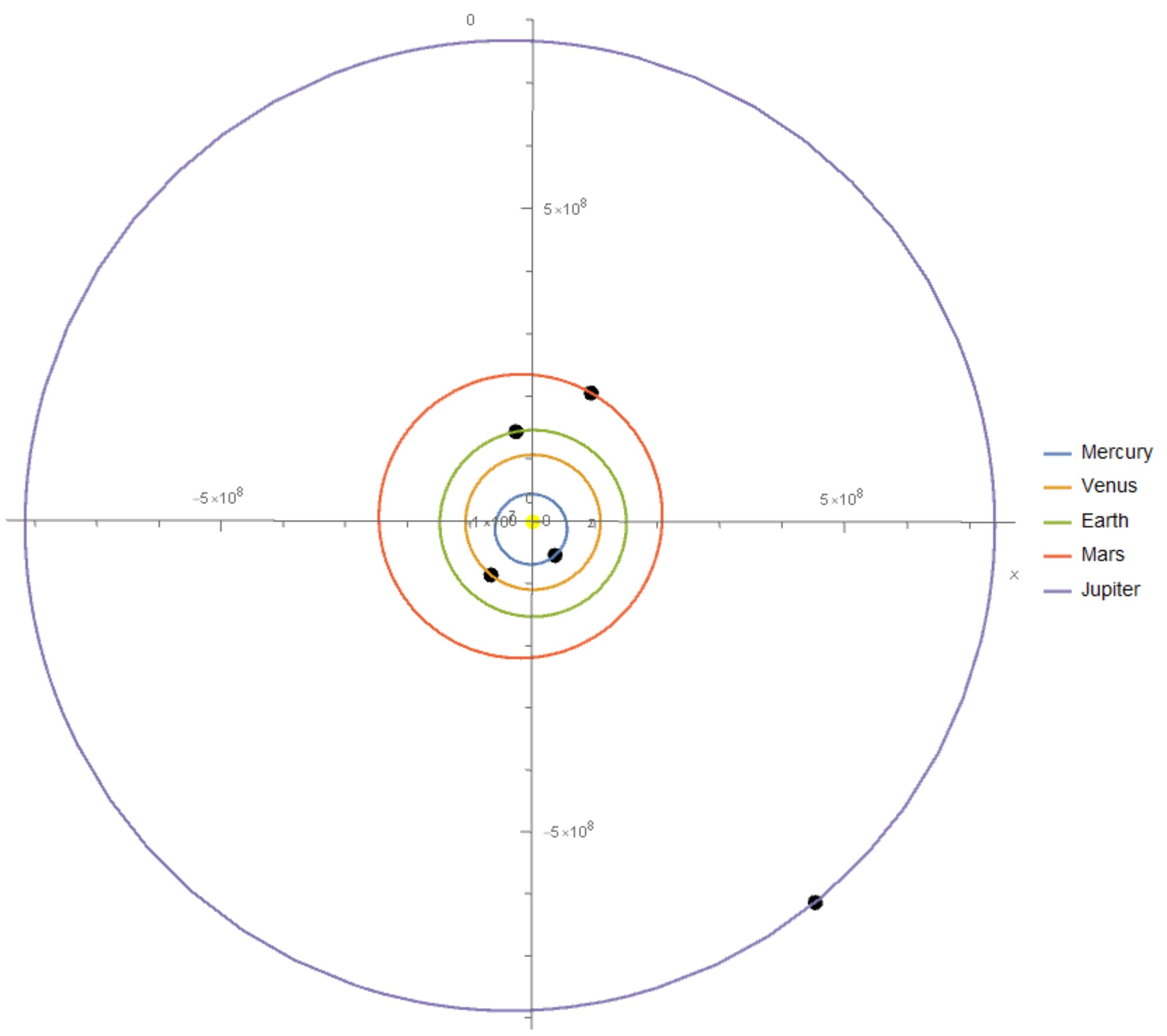}
  \caption{Top view: Positions on January 1, 2021}
\end{figure}

\end{document}